\documentclass[11pt]{article}
\usepackage{amsmath, amsthm, amscd, amsfonts, amssymb, graphicx, color}
\usepackage[utf8]{inputenc}
\usepackage{graphicx}
\usepackage{hyperref}
\usepackage[titlenumbered,ruled,noend,algosection]{algorithm2e}

\newtheorem{thm}{Theorem}[section]
\newtheorem{cor}[thm]{Corollary}
\newtheorem{obs}[thm]{Observation}
\newtheorem{lem}[thm]{Lemma}

\newtheorem{defn}[thm]{Definition}

\DontPrintSemicolon

\begin{document}
\title{Approximate Discontinuous Trajectory Hotspots}
\author{Ali Gholami Rudi\thanks{
{Department of Electrical and Computer Engineering},
{Bobol Noshirvani University of Technology}, {Babol, Iran}.
Email: {\tt gholamirudi@nit.ac.ir}.}}
\date{}
\maketitle
\begin{abstract}
A hotspot is an axis-aligned square of fixed side length $s$,
the duration of the presence of an entity moving in the plane in which is maximised.
An exact hotspot of a polygonal trajectory with $n$ edges can be found in $O(n^2)$.
Defining a $c$-approximate hotspot as an axis-aligned square of
side length $cs$, in which the duration of the entity's presence
is no less than that of an exact hotspot,
in this paper we present an algorithm to find a $(1 + \epsilon)$-approximate
hotspot of a polygonal trajectory with the time complexity
$O({n\phi \over \epsilon} \log {n\phi \over \epsilon})$,
where $\phi$ is the ratio of average trajectory edge length to $s$.
\end{abstract}


\section{Introduction}
\label{sint}
Many objects on earth move and huge collections of trajectory data have
been collected by tracking some of them with technologies like GPS devices.
In the analysis of these trajectories, many interesting geometric problems
arise, such as simplification \cite{kreveld18}, segmentation \cite{aronov16},
grouping \cite{buchin15}, classification \cite{alewijnse18}, and finding the
interesting regions like where objects spend a significant amount
of time \cite{benkert10,gudmundsson13,damiani14,rudi18}.

Few results have been published to present exact
geometric algorithms for the identification of regions that
are frequently visited, called \emph{hotspots} in the rest of this paper
(several heuristic algorithms have been published though, such as \cite{damiani14}).
The movement of an entity (its trajectory) is commonly represented as
a polygonal curve.  A set of vertices show the location of the entity
at specific points in time, and line segments (as edges) connect
contiguous vertices.
For multiple entities,
Benkert et al.~\cite{benkert10} defined a hotspot as
an axis-aligned square, which is visited by the maximum
number of distinct entities.
They presented an $O(n \log n)$ time sweep-line algorithm,
where $n$ is the number of trajectory vertices,
when only the inclusion of a trajectory vertex is considered a visit.
For the case where the inclusion of any portion of a trajectory
edge is a visit, they presented an $O(n^2)$ algorithm, which
subdivides the plane to find a hotspots.
They showed in both cases that their algorithm is optimal.

In a more recent paper, Gudmundsson et al.~\cite{gudmundsson13}
examined different definitions of trajectory hotspots, in
which the duration of the entity's presence is significant.
In this paper, we focus on one of their definitions, as follows:
a hotspot is an axis-aligned square of some pre-specified side length,
in which the entity (or entities) spends the maximum possible
duration and the presence of the entity in the region can
be discontinuous.  For this problem and for a trajectory with
$n$ edges, they presented an exact $O(n^2)$ algorithm,
which subdivides the plane based on the breakpoints
of the function that maps the location of a square of the
specified side length to the duration of the presence of the
entities in that square.

When $s$ is the side length of exact hotspots, a $c$-approximate
hotspot, where $c > 1$, is an axis-aligned square of side length $cs$,
in which the duration of the entity's presence is no less than
that of an exact hotspot.
In this paper we present an algorithm to find $(1 + \epsilon)$-approximate
hotspots of a trajectory in the plane.
The algorithm first subdivides each edge to small segments and then
finds the square that contains the maximum number of such segments.
The time complexity of the algorithm, which shall be presented in
the rest of this paper, is $O({n\phi \over \epsilon} \log {n\phi \over \epsilon})$,
where $\phi$ is the ratio of average length of trajectory edges to $s$.
We then use this algorithm to find a duration-approximate
hotspot of a trajectory $T$ with approximation ratio $1/4$,
i.e.~a square of side length $s$, in which the entity is present
at least $1/4$ of the time it is present in the exact hotspot.

This paper is organized as follows.
The algorithm and its analysis are presented in Section~\ref{smain},
after introducing the notation and defining some
of the concepts discussed in this paper in Section~\ref{sprel}.

\section{Preliminaries}
\label{sprel}
A trajectory describes the movement of an entity and is
represented as a sequence of vertices in the plane
with timestamps that specify the location of the entity at different
points in time.  The entity is assumed to move from one vertex to
the next in a straight line and with constant speed.
\begin{defn}
\label{dweight}
The weight of a square $r$ with respect to a trajectory $T$, denoted
as $w(r)$ is the total duration in which the entity spends inside it.
Also, $w(u)$ for any sub-trajectory (or edge) $u$ of $T$, indicates
the duration of $u$ (the difference between the timestamps of its
endpoints).
\end{defn}

The input to the problem studied in this paper is a
trajectory $T$ and the value of $s$.
The goal is to find a hotspot of $T$ (Definition \ref{dhotspot}).
Unless explicitly mentioned otherwise, every square discussed in
this paper is axis-aligned and has side length $s$.

\begin{defn}
\label{dhotspot}
A hotspot of trajectory $T$ in $R^2$ is a placement of a
square of side length $s$ in the plane with the maximum
weight (the duration of the presence the entity in the
square is maximised).
\end{defn}

A hotspot of a trajectory with $n$ edges can be found
with the time complexity $O(n^2)$~\cite{gudmundsson13}.
Our goal in this paper is finding approximate hotspots
of a trajectory more efficiently (Definition~\ref{dapprox}).

\begin{defn}
\label{dapprox}
A $c$-size-approximate (or $c$-approximate for brevity) hotspot is
a square whose weight is at least the weight of an exact hotspot
and its side length is $c$ times $s$.
\end{defn}

In Definition~\ref{dapprox}, hotspots are enlarged.
A more natural definition may be squares with the same size
as the exact ones, but with smaller weights (Definition~\ref{ddurapprox}).
\begin{defn}
\label{ddurapprox}
A $c$-duration-approximate hotspot is a square of side length
$s$, whose weight is at least $c$ times the weight of an exact hotspot.
\end{defn}

\section{The Approximation Algorithm}
\label{smain}
For a trajectory $T$ in $R^2$ and some constant $\epsilon$, where
$\epsilon > 0$, in this section we present an algorithm to find a
$(1 + \epsilon)$-approximate hotspot and use it to find a
$1/4$-duration-approximate hotspot.

\begin{figure}
	\centering
	\includegraphics{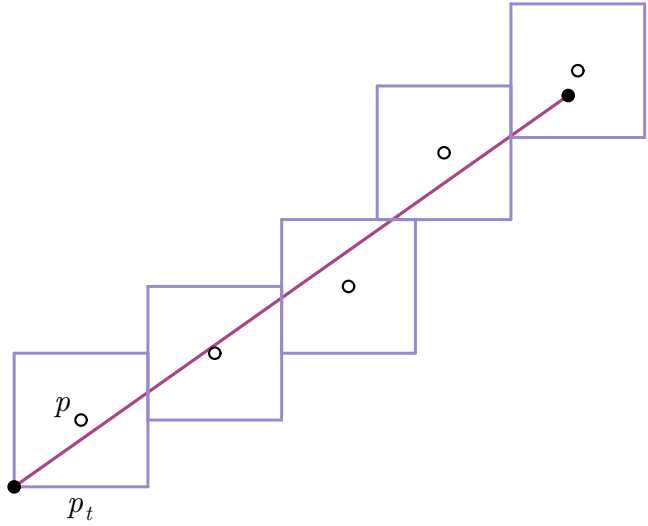}
	\caption{Tiling an edge and adding a point at the centre of each tile}
	\label{ftile}
\end{figure}
We first subdivide each edge of the trajectory
into segments of height and width at most $\epsilon s$; we
do so by covering each edge by non-overlapping
axis-aligned squares of side length $\epsilon s / 2$.
For each such tile, we add a point at its centre.
Let the weight of this point be the duration of the
portion of the edge that is inside its tile.
For every resulting point $p$, we use $p_t$ to denote its
corresponding tile and $p_g$ to denote the corresponding
segment (Figure~\ref{ftile}).
Note that the tiles of different edges may overlap.

\begin{defn}
\label{dpoint}
The point-weight of a square $r$ with respect to a trajectory $T$,
denoted as $w'(r)$, is the total weight of the points inside $r$.
Also, $w'(p)$ for point $p$ denotes the weight of point $p$.
\end{defn}

\begin{lem}
\label{linclusion}
Let $r$ be a square of side length $s$ and
let $r'$ be a square of side length $s + \epsilon s / 2$,
with the same centre of gravity.
We have $w(r) \le w'(r')$.
\end{lem}
\begin{proof}
Every point in the sub-trajectory inside $r$ belongs to some
segment $p_g$ corresponding to tile $p_t$ and point $p$.
Let $P$ be the set of all points, whose segments are intersected
by $r$.
Since $r$ contains all or some part of any segment corresponding
to these points, $w(r) \le \sum_{p \in P} w(p)$.
Since $r'$ is $\epsilon s / 4$ longer than $r$ at each side,
when $r$ intersects $p_t$, the centre of $p_t$, $p$, is contained
in $r'$.  Therefore, $\sum_{p \in P} w(p) \le w'(r')$,
which implies that $w(r) \le w'(r')$.
\end{proof}

\begin{lem}
\label{lweight}
Let $r$ be a square of side length $s + \epsilon s / 2$ and
let $r'$ be a square of side length $s + \epsilon s$,
with the same centre of gravity.
We have $w'(r) \le w(r')$.
\end{lem}
\begin{proof}
Let $P$ denote the set of points inside $r$; the point-weight
of square $r$ is the sum of the weights of the points in $P$.
Suppose $p$ is a member of $P$.
Since $p$ is the centre of $p_t$ and $p$ is inside $r$,
the whole of $p_t$ is contained in $r'$, because $r'$ is
$\epsilon s/4$ longer than $r$ at each side.
This implies that the whole of $p_g$ is inside $r'$.
Therefore, $\sum_{p \in P} w'(p) = \sum_{p \in P} w(p_g) \le w(r')$,
as required.
\end{proof}

In Theorem~\ref{tsweep}, we use a data structure for storing $m$
numbers that supports obtaining their maximum in $O(1)$ and
increasing the numbers in any continguous interval of the numbers
by a value in $O(\log m)$.
This can be implemented by augmenting Range Minimum Query (RMQ)
data structures, like the Fenwick tree \cite{fenwick96}, and
by storing changes to the leaves of subtrees at internal nodes,
instead of updating the value of every node in that subtree.

\begin{thm}
\label{tsweep}
Given a trajectory $T$ in $R^2$ and the value of $s$,
after tiling, a square of side length $s + \epsilon s / 2$ and
with the maximum point-weight can be found with the time
complexity $O(m \log m)$, where $m$ is the number of points.
\end{thm}
\begin{proof}
To find a square with the maximum weight, it suffices to search
among the squares that have a point on each of their
lower and left sides (any square with the maximum weight can
be moved up and right without changing its point-weight, until
their lower and left sides meet a point).

Let $\sigma$ be the sequence of points in $P$,
ordered by their $y$-coordinate.
We sweep the plane horizontally using two parallel sweep lines with
distance $s + \epsilon s / 2$ as follows.
During the sweep line algorithm, we maintain the point-weight of $m$
squares in the data structure $W$, such that the $i$-th number in
$W$ denotes the point-weight of the square
whose lower side has the same height as the $i$-th point and
its left and right sides are on the sweep lines; we use $r_i$
to refer to this square.

In the sweep line algorithm, we process the following events:
when the left or the right sweep line intersects a point $p$.
We process an event for point $p$ as follows.
Let $p$ be the $i$-th item of $\sigma$ and
let $j$ be the index of the lowest point in $\sigma$ such that
the difference between the height of $p$ and the $j$-th
point of $\sigma$ is at most $s + \epsilon s / 2$;
the value of $j$ can be found using binary search on $\sigma$.
When $p$ meets the right sweep line, we increase the
weight of every square $r_k$ such that $j \le k \le i$
by $w'(p)$, because every such square contains $p$.
Similarly, when $p$ meets the left sweep line, we decrease the
weight of every square $r_k$ such that $j \le k \le i$
by $w'(p)$.
During the sweep line algorithm, we record the square with
the maximum weight in $W$.  At the end of the algorithm, it
denotes a square with the maximum point-weight among all
squares with a vertex on their lower and left sides.

The complexity of sorting $m$ points based on their $y$-coordinate
and handling $m$ events, each with complexity $O(\log m)$ is
$O(m \log m)$.
\end{proof}

In Theorem~\ref{tmain}, we present and analyse the main algorithm.

\begin{thm}
\label{tmain}
Given a trajectory $T$ in $R^2$ and the value of $s$,
there is an algorithm that finds a $(1 + \epsilon)$-approximate
hotspot of trajectory $T$ with the time complexity
$O({n\phi \over \epsilon} \log {n\phi \over \epsilon})$,
in which $\phi$ is the ratio of average length of trajectory edges to $s$.
\end{thm}
\begin{proof}
After tiling, as described in the beginning of this section,
an edge of length $d$ is subdivided into at most
$\lceil {d \over {\epsilon s}} \rceil$ segments.
Therefore, if the total length of the edges of $T$ is $a$,
the number of resulting segments is at most ${a \over {\epsilon s}} + n$,
which is equal to $O({n\phi \over \epsilon})$ asymptotically.
Theorem~\ref{tsweep} shows how the square with the maximum
point-weight, $r$, can be found in $O(m \log m)$.
Let $r'$ be the square with the same centre of gravity as $r$ but
of side length $s + \epsilon s$.  Also, let $h$ denote the weight
of an exact hotspot of $T$.
We show that $r'$ is a $(1 + \epsilon)$-approximate hotspot.

Lemma~\ref{linclusion} implies that there is at least one square with
side length $s + \epsilon s / 2$ whose point-weight is equal to
$h$, the weight of an exact hotspot of $T$ (of side length $s$).
Since, $r$ is the square with the maximum point-weight among
squares of side length $s + \epsilon s / 2$, its point-weight
is at least $h$.
Furthermore, Lemma~\ref{lweight} shows that the weight of $r'$ is at
least the point-weight of $r$ (at least $h$).
Therefore, the algorithm finds a square of side length
$s + \epsilon s$ and with weight at least $h$;
a $(1 + \epsilon)$-approximate hotspot by Definition~\ref{dapprox}.
\end{proof}

We now use the algorithm presented in Theorem~\ref{tmain} to
find a duration-approximate hotspot of a trajectory in the
plane.  For that, we need Observation~\ref{ocorners}, which
can be shown by placing the smaller squares at the corners
of a hotspot.

\begin{obs}
\label{ocorners}
Let $h$ be the weight of an exact hotspot of a trajectory $T$
in the plane.  There exists a square of side length $cs$ and
weight at least $h / 4$, provided that $c \ge 1/2$.
\end{obs}

\begin{cor}
\label{cmain}
Given a trajectory $T$ in $R^2$ and the value of $s$,
there is an algorithm that finds a $1/4$-duration-approximate
hotspot of trajectory $T$ with the time complexity $O({n\phi} \log {n\phi})$,
in which $\phi$ is the ratio of average length of trajectory edges to $s$.
\end{cor}
\begin{proof}
Theorem~\ref{tmain} for hotspot side length $s' = s / 2$ and $\epsilon = 1$
yields a square $r$ of weight $h$ and side length $s$.
Since $h$ is the maximum weight of the squares with side length $s/2$,
Observation~\ref{ocorners} implies that the weight of an exact hotspot
of side length $s$ of $T$ cannot be greater than $4h$.
Therefore, $r$ is a $1/4$-duration-approximate hotspot of $T$.
\end{proof}


\end{document}